\documentclass[letterpaper, 10 pt, conference]{ieeeconf} 

\IEEEoverridecommandlockouts 

\usepackage{amssymb, amsthm, amsmath, latexsym, amsfonts, mathrsfs, amsbsy,mathtools,relsize}
\usepackage[T1]{fontenc}
\usepackage[dvips]{graphicx}
\usepackage[x11names]{xcolor}
\usepackage{listings}
\usepackage{enumerate}
\usepackage{subcaption}
\usepackage{soul}
\usepackage{pgfplots}
\usepackage{tikz}
\usepackage{float}
\usepackage{dsfont}
\usepackage{fontawesome5}
\usepackage{algorithm}
\usepackage[algo2e]{algorithm2e} 
\usepackage[noend]{algpseudocode}
\usepackage[colorlinks]{hyperref}
\hypersetup{
    citecolor=Firebrick4,
    linkcolor=Firebrick4,   
    urlcolor=Red3}
\usepackage{cite}

\definecolor{RdYlBu1}{HTML}{D73027}
\definecolor{RdYlBu2}{HTML}{FC8D59}
\definecolor{RdYlBu3}{HTML}{FEE090}
\definecolor{RdYlBu4}{HTML}{FFFFBF}
\definecolor{RdYlBu5}{HTML}{E0F3F8}
\definecolor{RdYlBu6}{HTML}{91BFDB}
\definecolor{RdYlBu7}{HTML}{4575b4}

\definecolor{RdYlGr1}{HTML}{d73027}
\definecolor{RdYlGr2}{HTML}{fc8d59}
\definecolor{RdYlGr3}{HTML}{fee08b}
\definecolor{RdYlGr4}{HTML}{ffffbf}
\definecolor{RdYlGr5}{HTML}{d9ef8b}
\definecolor{RdYlGr6}{HTML}{91cf60}
\definecolor{RdYlGr7}{HTML}{1a9850}

\usetikzlibrary{shapes,shapes.geometric,arrows,arrows.meta,fit,calc,positioning,automata,through,intersections}
\tikzset{diamond state/.style={draw,diamond}}

\newcommand{\nnint}{\mathbb{Z}_{\geq 0}}

\newtheoremstyle{theoremdd}
  {\topsep}
  {\topsep}
  {\itshape}
  {0pt}
  {\bfseries}
  {.}
  { }
  {\thmname{#1}\thmnumber{ #2}\textnormal{\thmnote{ (#3)}}}

\theoremstyle{theoremdd}
\newtheorem{theorem}{Theorem}
\newtheorem{lemma}{Lemma}
\newtheorem{assumption}{Assumption}

\newtheorem{remark}{Remark}

\newcommand{\set}[1]{\mathcal{#1}} 
\newcommand{\ie}{\textit{i.e.,~}} 
\newcommand{\eg}{\textit{e.g.,~}} 
\newcommand{\inneighbor}[1]{\set{N}_{#1}^{\texttt{in}}}
\newcommand{\outneighbor}[1]{\set{N}_{#1}^{\texttt{out}}}
\newcommand{\indegree}[1]{d_{#1}^{\texttt{in}}}
\newcommand{\outdegree}[1]{d_{#1}^{\texttt{out}}}

\newcommand{\vect}[1]{\mathbf{#1}} 
\newcommand{\grvect}[1]{\boldsymbol{#1}} 

\newenvironment{list4}{
\begin{list}{$\bullet$}{%
    \setlength{\itemsep}{0.05cm}
    \setlength{\labelsep}{0.2cm}
    \setlength{\labelwidth}{0.3cm}
    \setlength{\parsep}{0in} 
    \setlength{\parskip}{0in}
    \setlength{\topsep}{0in} 
    \setlength{\partopsep}{0in}
    \setlength{\leftmargin}{0.18in}}}
{\end{list}}

\title{\LARGE \bf
{Multi-cluster distributed optimization in open multi-agent systems\\ over directed graphs with acknowledgement messages}}

\author{
Evagoras Makridis$^{1,*}$, Gabriele Oliva$^{2}$, and Themistoklis Charalambous$^{1,3}$
\thanks{$^1$Department of Electrical and Computer Engineering, School of Engineering, University of Cyprus, Nicosia, Cyprus.}
\thanks{$^2$Department of Engineering, University Campus Bio-Medico of Rome, Via Alvaro del Portillo, 21 - 00128 Roma, Italy.} 
\thanks{$^3$Department of Electrical Engineering and Automation, School of Electrical Engineering, Aalto University, Espoo, Finland.}
\thanks{$^*$Corresponding author. Email: {\tt makridis.evagoras@ucy.ac.cy}.}
\thanks{This work has been partly funded by MINERVA, a European Research Council (ERC) project funded under the European Union's Horizon 2022 research and innovation programme (Grant agreement No. 101044629).
}}

\begin{document}

\maketitle

\begin{abstract}
In this paper, we tackle the problem of distributed optimization over directed networks in open multi-agent systems (OMAS), where agents may dynamically join or leave, causing persistent changes in network topology and problem dimension. These disruptions not only pose significant challenges to maintaining convergence and stability in distributed optimization algorithms, but could also break the network topology into multiple clusters, each one associated with its own set of objective functions. To address this, we propose a novel Open Distributed Optimization Algorithm with Gradient Tracking (\textsc{Open-GT}), which employs: \emph{(a)} a dynamic mechanism for detecting active out-neighbors through acknowledgement messages, and \emph{(b)} a fully distributed max-consensus procedure to spread information regarding agent departures, in possibly unbalanced directed networks. We show that when all active agents execute \textsc{Open-GT}, the optimization process in each formed cluster remains consistent, while the agents converge to their cluster-wide optimal solution if there exists a time after which the network remains unchanged. Finally, we validate our approach in a simulated environment with dynamically changing agent populations, demonstrating its resilience to network variations and its ability to support distributed optimization under OMAS dynamics.
\end{abstract}

\begin{keywords}
distributed optimization, open multi-agent systems, OMAS, multi-cluster optimization, gradient tracking
\end{keywords}

\section{Introduction}\label{sec:introduction}
Distributed optimization in multi-agent systems (MAS) has attracted significant interest due to applications in sensor networks, robotic coordination, and large-scale machine learning. Agents collaborate to solve a global optimization problem using only local information exchanged with neighbors. Traditional methods typically assume a fixed number of agents and model communication via static undirected \cite{nedic2009distributed} or directed graphs \cite{nedic2014distributed,nedic2017achieving,xi2017add}. However, in realistic scenarios, the agent population changes over time, motivating the shift towards open MAS (OMAS), where agents may join or leave dynamically at runtime.

This dynamic setting introduces several challenges. In particular, while existing works often require persistent network connectivity \cite{hendrickx2020stability,hsieh2021optimization,hayashi2022distributed,vizuete2022resource,de2024random,deplano2025optimization}, in practice, global connectivity may be unnecessary. Instead, agent arrivals and departures naturally induce multiple clusters, each forming a connected sub-network with its own active agents and local minimizer. Thus, a practical focus shifts to ensuring connectivity within each cluster, allowing for flexible and scalable distributed optimization in dynamic OMAS.

\begin{figure}[t]
  \begin{subfigure}[b]{0.5\linewidth}
    \centering
    \includegraphics[scale=0.95]{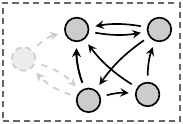} 
    \caption{Initial network} 
    \label{fig:a} 
    \vspace{2ex}
  \end{subfigure}
  \begin{subfigure}[b]{0.5\linewidth}
    \centering
    \includegraphics[scale=0.95]{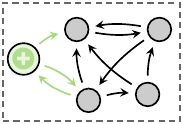}    
    \caption{Agent arrival} 
    \label{fig:b} 
    \vspace{2ex}
  \end{subfigure} 
  \begin{subfigure}[b]{0.5\linewidth}
    \centering
    \includegraphics[scale=0.95]{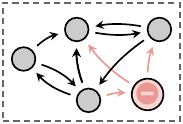}
    \caption{Agent departure} 
    \label{fig:c} 
  \end{subfigure}
  \begin{subfigure}[b]{0.5\linewidth}
    \centering
    \includegraphics[scale=0.95]{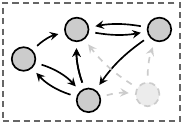}    
    \caption{Stable network} 
    \label{fig:d} 
  \end{subfigure} 
  \caption{Snapshots of an OMAS with directed communication links in which agents join the network, and present ones depart from the network as time progresses.}
  \label{fig:omas_diagram}
\end{figure}

In realistic OMAS settings, agents often experience asymmetric communication due to network constraints or variations in transmission power. While prior works have addressed OMAS under undirected communication \cite{hsieh2021optimization,hayashi2022distributed,liu2024distributed,deplano2025optimization}, such assumptions limit their applicability in dynamic environments where agents may join or leave during runtime.
In contrast, directed communication offers improved scalability and adaptability, but introduces new challenges for algorithm design. Existing results for directed networks largely focus on the distributed average consensus problem; see, \eg \cite{makridis2025average,hadjicostis2025distributed}.
More recently, \cite{sawamura2024distributed} considered distributed optimization in OMAS with directed communication. Although this approach deals with distributed optimization in directed networks of OMAS, it assumes that the active OMAS remains strongly connected throughout the optimization process which might be restrictive in open networks where agents may depart. Additionally, \cite{sawamura2024distributed} presumes agents have prior knowledge of their active out-neighbors, which can be a strong requirement in networks with directional information flow.

In dynamically evolving, asymmetric networks, ensuring consistent information flow is particularly challenging, as agents may become disconnected, limiting their ability to collaboratively optimize towards the true global objective of the active OMAS. Furthermore, when agents are unaware of changes in their out-neighborhood, \eg agent departures, they cannot assign correct communication weights on their outgoing links, leading to incorrect (biased) convergence. Moreover, disconnections may split the network into multiple isolated clusters, within which agents are both unable and no longer obligated to coordinate globally. Consequently, optimization efforts become localized, with each cluster of agents independently optimizing its own objective function rather than collectively optimizing toward the global objective of the overall OMAS.

Towards this direction, we introduce the following contributions. We propose a novel framework for distributed optimization in OMAS that supports dynamic cluster formation, allowing agent departures and arrivals to split or merge the network under directed and time-varying communication links, as illustrated in Fig.~\ref{fig:omas_diagram}. More specifically, we propose a distributed optimization algorithm for clustered OMAS, hereinafter referred to as \textsc{Open-GT}, that optimizes cluster-wide objective functions, while incorporating the exchange of acknowledgment messages to detect local topology changes. These messages enable agents to identify topology changes via max-consensus iterations performed between optimization rounds. Upon detecting such changes, agents can reassign the weights on their outgoing links, ensuring correctness on their optimization variables which are driven towards the new minimizer, despite changes in the network's composition. We support these contributions with theoretical analysis and simulation results, demonstrating the correctness and effectiveness of \textsc{Open-GT} in OMAS settings.


\section{Preliminaries}\label{sec:background}

\subsection{Mathematical Notation}
The sets of real, integer, and natural numbers are denoted as $\mathbb{R}, \mathbb{Z}$, and $\mathbb{N}$, respectively. 
The set of nonnegative integer (real) numbers is denoted as $\mathbb{Z}_{\geq 0}$ ($\mathbb{R}_{\geq 0}$). Matrices are denoted by capital letters, and vectors by small boldface letters. The transpose of a matrix $A$ and a vector $\vect{x}$ are denoted as $A^\top$ and $\vect{x}^\top$, respectively. The all-ones and all-zeros vectors are denoted by $\mathbf{1}$ and $\mathbf{0}$, respectively, with their dimensions being inferred from the context.

\subsection{Open Network Model}\label{sec:open_network_model} 
In this work, we assume that agents are able to join and leave the network at will. For simplicity of exposition, we assume a finite set $\set{V}=\{v_1, \cdots, v_n\}$ that captures all $n=|\set{V}|$ agents potentially active to the network. However, at each time step $k$, only a subset of the agents, denoted by $\set{V}_k \subseteq  \set{V}$, is {\em active}. For analysis purposes, we denote the activation of the agents by a time-varying indicator vector $\grvect{\alpha}_k := \begin{bmatrix}
    \alpha_{1,k},\ldots,\alpha_{n,k}\end{bmatrix}^{\top} \in \{ 0, 1 \}^n$, where $\alpha_{j,k}=1$ if agent $v_j\in\set{V}$ is active at time step $k$ (\ie if $v_j\in \mathcal{V}_k$), while $\alpha_{j,k}=0$, otherwise. Based on this activation vector, the interconnection topology of the active agents in the network is modeled by a (possibly unbalanced) time-varying digraph $\set{G}_k=(\set{V}_k,\set{E}_k)$. The number of agents active in the network at time $k$ is denoted by $n_k=|\set{V}_k|$. The interactions between active agents are captured by the set of digraph edges $\set{E}_k \subseteq \set{V}_k \times \set{V}_k$. A directed edge denoted by $\varepsilon_{ji} \in \set{E}_k$ indicates that node $v_j$ receives information from node $v_i$ at time step $k$. The nodes from which node $v_j$ receives information at time step $k$ are called in-neighbors of node $v_j$ at time step $k$, and belong to the set $\inneighbor{j,k}=\{v_i \in \set{V}_k| \varepsilon_{ji} \in \set{E}_k\}$. The number of nodes in this set is called the current in-degree of node $v_j$ and is denoted by $\indegree{j,k} = |\inneighbor{j,k}|$. The nodes that receive information from node $v_j$ directly at time step $k$ are called out-neighbors of node $v_j$ at time step $k$, and belong to the set $\outneighbor{j,k}=\{v_l \in \set{V}_k| \varepsilon_{lj} \in \set{E}_k\}$. The number of nodes in this set is called the current out-degree of node $v_j$ and is denoted by $\outdegree{j,k}= |\outneighbor{j,k}|$. The potential out-neighbors of node $v_j$ is given by $\outneighbor{j} = \cup_{k=0}^{\infty} \outneighbor{j,k}$. Clearly, each node $v_j \in \set{V}_k$ has immediate access to its own local state.

The structure of the network at each time instant $k\in\nnint$ can be formally defined by the following three subsets that distinguish the operating modes of the agents. 
\begin{list4}
    \item Agents that are active in the network at time $k-1$ and they are still active at time $k$, belong to the set of \emph{remaining agents}, \ie
        \begin{align}
            \set{R}_k = \set{V}_{k-1} \cap \set{V}_{k}.    
        \end{align}
    \item Agents that are inactive at time $k-1$ but are active at time $k$, belong to the set of \emph{joining agents}, \ie
        \begin{align}
        \set{J}_k = \set{V}_{k} \setminus \set{R}_{k}.    
        \end{align}
    \item Agents that are active at time $k-1$ but are inactive at time $k$, belong to the set of \emph{departing agents}, \ie
        \begin{align}
        \set{D}_k = \set{V}_{k-1} \setminus \set{R}_{k}.
        \end{align}
\end{list4}
Notice that, the set of active agents at time $k+1$ is given by $\set{V}_{k+1}= \{ \set{V}_k \setminus \set{D}_{k+1}\} \cup \set{J}_{k+1}$. In general, the set of times when agent $v_j\in\set{V}$ arrives in the network is given by
\begin{align}
    \set{K}_j^{\text{arr}} = \{k\in \nnint \mid v_j \in \set{J}_k\}.
\end{align}
The set of times when agent $v_j\in\set{V}$ departs from the network is given by
\begin{align}
    \set{K}_j^{\text{dep}} = \{k\in \nnint \mid v_j \in \set{D}_k\}.
\end{align}

\subsection{Problem Formulation}
We consider an OMAS whose network composition at time $k \in \nnint$ is represented by a time-varying directed graph $\set{G}_k = (\set{V}_k, \set{E}_k)$. Due to the system’s openness, disjoint clusters $\set{G}_k^q \in \set{Q}$ (with $\set{Q}$ a finite set) may form due to agent departures and arrivals. Each cluster $\set{G}_k^q = (\set{V}_k^q, \set{E}_k^q)$ corresponds to a strongly connected component of the maximal digraph $\set{G}=(\set{V},\set{E})$, where $\set{E}:=\cup_{k=0}^{\infty} \set{E}_k$ is the set of all potential interactions. The set of active times of the cluster $q\in\set{Q}$ is denoted by $\set{K}^q = \{ k \in \nnint \mid \set{V}_k^q \subseteq \set{V}_k \}$. Each cluster has $|\set{V}_k^q|$ agents, with all clusters being disjoint at any $k \in \nnint$, \ie, $\set{V}_k^q \cap \set{V}_k^{q'} = \emptyset$ for $q \neq q'$ and $\cup_{q\in\set{Q}}\set{V}_k^q=\set{V}_k$. Here it is important to note that agents interact only with their immediate neighbors without being aware of the cluster they belong to.

Within this setting, we study a distributed optimization problem in which active agents $v_j \in \set{V}_k$ aim to compute an optimal solution $x_k^{q\ast} \in \mathbb{R}^d$ that minimizes the sum of local cost functions in their current cluster. Each agent $v_j$ holds a possibly time-varying convex local cost function $f_{j,k}: \mathbb{R}^d \to \mathbb{R}$, allowing agents to rejoin the system with updated objectives and new initial conditions. Agents maintain local estimates $x_{j,k} \in \mathbb{R}^d$ and update them iteratively to converge to $x_k^{q\ast}$. That is, agents in cluster $q$ aim to reach consensus on a common minimizer of the sum of their active local costs, \ie
\begin{align}\label{eq:optimization_problem}
    \operatorname*{minimize}\limits_{x\in\mathbb{R}^d} \ f_k^q(x) = \sum_{v_j\in\set{V}_k^q} f_{j,k}(x), 
\end{align}
where $f_k^q$ denotes the sum of local cost functions that belong to cluster $q$ at time $k$.

\begin{remark}
    Clearly, when the topology of the OMAS can only form one cluster $|\set{Q}|=1$, the problem in \eqref{eq:optimization_problem} reduces to the classical distributed optimization problem in OMAS. If a \emph{trivial~cluster} exists, \ie $|\set{V}_k^q|=1$, then the cluster (the agent) solves a local optimization problem based on its local variables and without the need of communication.
\end{remark}

\subsection{Assumptions}
In order to present the proposed distributed algorithm and its analysis, let us now state some assumptions.

\begin{assumption}{\textnormal{(Lipschitz gradient)}}\label{ass:1}
The local function $f_{j,k}$ of agent $v_j\in\set{V}$ is continuously differentiable and has $L_j$-
Lipschitz continuous gradients for all $k\in\mathbb{Z}_{\geq0}$, that is, there exists a constant $L_{j}>0$ such that
\begin{align}
    \|\nabla f_{j,k}(x)-\nabla f_{j,k}(y)\|\leq L_{j}\|x-y\|,
\end{align}
for all $x,y\in\mathbb{R}^d$.
\end{assumption}

\begin{assumption}{\textnormal{(Strong convexity)}}\label{ass:2}
The global objective function $f_k=\sum_{j\in\set{V}} f_{j,k}$ is $\mu$-strongly convex for some $\mu>0$, and all $k\in\nnint$, \ie
\begin{align}
    f_k(y) \geq f_k(x)+\nabla f_k(x)^{\top}(y-x)+\frac{\mu}{2}\|y-x\|_2^2,
\end{align}
for all $x,y\in\mathbb{R}^d$.
\end{assumption}

\begin{assumption}{\textnormal{($\beta$-strongly connected cluster sequence).}}\label{ass:3}
    The sequence of time-varying sub-digraphs which model the agents' interactions within cluster $q$, \ie $\set{G}_k^q=(\set{V}_k^q,\set{E}_k^q)$ is $\beta$\emph{-strongly~connected}, \ie there exists an integer $\beta>0$ such that for all $k\in\set{K}^q$, the union of sub-digraphs
    \begin{align}
        \set{G}_{k:\beta}^q:= \Bigl\{ \set{V}_k^q, \cup_{t=k\beta}^{(k+1)\beta-1} \set{E}^q_t \Bigr\},
    \end{align}
    is strongly connected.
\end{assumption}

Assumption~\ref{ass:3} relaxes the strong connectivity condition used in \cite{sawamura2024distributed}, allowing agents to form disconnected clusters and minimize their cluster-wide objective. For fixed-agent sets with time-varying links, \ie $\set{G}_k = (\set{V}, \set{E}_k)$, such assumption is standard in distributed optimization \cite{nedic2014distributed, nguyen2023accelerated, li2024accelerated, wang2023distributed}.

\begin{assumption}{\textnormal{(Acknowledgement messages)}}\label{ass:4}
At the beginning of each time step $k\in\nnint$, each active agent $v_j\in\set{V}_k$ can send an $1$-bit message $\alpha_{j,k}=1$ to its in-neighbors $v_i\in\inneighbor{j}$, which is only received by the currently active ones $v_i\in\inneighbor{j,k}$.
\end{assumption}

Assumption~\ref{ass:4}, inspired by protocols such as TCP, ARQ/HARQ, and ALOHA, models $1$-bit acknowledgment messages exchanged over dedicated, narrowband control channels that are distinct from the directed data channels. Specifically, for each directed data channel, a corresponding reverse control channel is used to carry these acknowledgments. Over such channels, agents $v_j \in \set{V}_k$ exchange $1$-bit acknowledgments \cite{makridis2023utilizing,makridis2023harnessing} to detect their active out-neighbors $\outneighbor{j,k}$ and compute $\outdegree{j,k} = |\outneighbor{j,k}|$, enabling proper weight assignment $c_{lj,k}$ for information exchange.

\section{\textsc{Open-GT} Algorithm Development}
In this section, we present the \textsc{Open-GT} algorithm which enables agents communicating over directed links to track the minimizer of the cluster they belong to. Each active agent at time $k \in \nnint$ maintains the variables\footnote{We consider scalar state variables for clarity, though the algorithm extends to vector-valued states.} $x_{j,k}, y_{j,k}, w_{j,k} \in \mathbb{R}$ and $h_{j,k} \in \{0,1\}$. The variable $x_{j,k}$ is updated via a gradient step on the local cost function $f_{j,k}(x_{j,k})$ and represents the current estimate of the optimal solution of \eqref{eq:optimization_problem} as perceived by agent $v_j \in \set{V}_k$. The variable $y_{j,k}$ supports consensus by compensating for imbalance in directed communication, while $w_{j,k}$ tracks the gradient of the global objective. The binary variable $h_{j,k}$ is determined through max-consensus rounds between consecutive optimization steps to trigger a reset of the variables of agent $v_j$ when a departure is detected in the network.

\textbf{Arriving mode:}
Each arriving agent $v_j\in\set{J}_{k+1}$ initializes its variables to 
\begin{subequations}\label{eq:opengt_init}
\begin{align}
    x_{j,k+1}&=\widehat{x}_{j,k+1},\\
    y_{j,k+1}&=1,\\
    w_{j,k+1}&=\nabla f_{j,k+1}(x_{j,k+1}),\\
    h_{j,k+1}&=0,
\end{align}    
\end{subequations}
where $\widehat{x}_{j,k+1}\in\mathbb{R}$ denotes its initial estimate regarding the optimal solution when it (re)-joins the network. Moreover, it determines the gradient step-size $\gamma$ which is assumed to be common for all agents in the network, and sets the number of max-consensus iterations $\Lambda=\overline{D}$, where $\overline{D}\geq D$ is an upper-bound on the diameter\footnote{The diameter of the network $D$ is defined as the longest shortest path between any two nodes (respecting the direction of the links).} $D$ of the maximal network.   

\textbf{Remaining mode:} %
At each optimization round $k$, each remaining agent $v_j \in \set{R}_k$ performs $\Lambda$ max-consensus iterations on $h_{j,k}$ to identify its current out-degree and reassign the weights $c_{lj,k} \geq 0$ on its out-going links. After these inner iterations, each agent broadcasts the scaled values $c_{lj,k}( x_{j,k} - \gamma w_{j,k})$, $c_{lj,k}y_{j,k}$, and $c_{lj,k}w_{j,k}$. The \emph{push weights} $c_{lj,k}$ are computed before transmission and typically depend on the agent's out-degree (see \cite{hadjicostis2015robust,charalambous2015distributed} for estimation methods, and \cite{makridis2023utilizing} for feedback-based computation in closed MAS). In OMAS, where out-degrees may vary, a distributed and dynamic mechanism is essential to enable agents track the number of their currently active out-neighbors. Before presenting this mechanism, we formally describe the information exchange and local update procedure for each $v_j \in \set{R}_{k+1}$ under the \textsc{Open-GT} algorithm.

At round $k+1$, each remaining agent $v_j\in\set{R}_{k+1}$ receives the values $c_{ji,k} (x_{i,k} -\gamma w_{i,k})$, $c_{ji,k} y_{i,k}$, and $c_{ji,k} w_{i,k}$ from its currently active in-neighbors $v_i\in\inneighbor{j,k}$, and updates, \ie
\begin{subequations}\label{eq:opengt_update}
\begin{align}
\hspace{-5pt}
    x_{j,k+1} &= \sum_{v_i \in \mathcal{N}_{j,k}^\texttt{in+}} \!\! c_{ji,k} (x_{i,k} -\gamma w_{i,k}), \label{eq:opengt_update_x}\\
\hspace{-5pt}
    y_{j,k+1} &= \sum_{v_i \in \mathcal{N}_{j,k}^\texttt{in+}} \!\! c_{ji,k} y_{i,k}, \label{eq:opengt_update_y}\\
\hspace{-5pt}
    z_{j,k+1} &= x_{j,k+1}/y_{j,k+1}, \label{eq:opengt_update_z}\\
\hspace{-5pt}
    w_{j,k+1} &= g_{j,k+1} \!+\! (1{-}h_{j,k}) \! \Bigg(\! \sum_{v_i \in \mathcal{N}_{j,k}^\texttt{in+}} \!\!\! c_{ji,k} w_{i,k} - g_{j,k} \! \Bigg),\label{eq:opengt_update_w}
\end{align}
\end{subequations}
where $g_{j,k} := \nabla f_{j,k}(z_{j,k})$, and $\mathcal{N}_{j,k}^\texttt{in+}:=\inneighbor{j,k} \cup\{v_j\}$.

\textbf{Dynamic weight assignment:}
In directed networks, remaining agents $v_j \in \set{R}_{k+1}$ may continue sending information to departed out-neighbors $v_l \in \outneighbor{j} \cap \set{D}_{k+1}$ since they are not aware about their departure. This leads to incorrect scaling, as the push weight $c_{lj,k}$ is computed using an outdated out-degree $\outdegree{j,k}$. To address this, we design a mechanism for dynamically detecting active out-neighbors $\outneighbor{j,k}$. Specifically, in each optimization round $k$, before transmitting any variables, each remaining agent $v_j \in \set{R}_k$ receives one-bit acknowledgment messages $\alpha_{l,k} = 1$ from its currently active out-neighbors $v_l \in \outneighbor{j}$, which are used to update its out-degree and assign accurate weights $c_{lj,k}$ as
\begin{align}\label{eq:c-weights}
    c_{lj,k}=\begin{cases}
    \hfil \frac{1}{1+\outdegree{j,k}},\!\!\!& \text{ if } \alpha_{l,k}=1 \lor l=j,\\
    \hfil 0, & \text{ otherwise},
    \end{cases}
\end{align}
where its current out-degree can be computed by summing the received acknowledgement messages over its potential out-neighbors, \ie $\outdegree{j,k}=\sum_{v_l\in\outneighbor{j}} \alpha_{l,k}$.

Although this mechanism would allow all the agents to properly assign their weights, this would not be enough for the algorithm in \eqref{eq:opengt_update} to work. The reason is due to the fact that agents within the same cluster need to know whether there is a departure in the cluster, even if the departure does not involve an agent from their out-neighborhood.  

\textbf{Coordinating departure information:} A missing acknowledgment message signals the departure of an out-neighbor of agent $v_j$. By exchanging this information and performing max-consensus iterations between consecutive optimization rounds, all active agents within the same cluster can reach agreement on the departure event, \ie determine $h_{j,k}$, across the cluster. In the light of \eqref{eq:opengt_update}, $h_{j,k}=1$ triggers a coordinated reset of the gradient tracking variables to $w_{j,k}=\nabla f_{j,k}(z_{j,k})$, for the remaining agents that belong in the same cluster $v_{j,k}\in\set{R}_{k}\cap \set{V}_{k}^q$.

In what follows, we design a fully distributed mechanism capable to spread the information about missing acknowledgement messages (corresponding to departing agents) within the active clusters via max-consensus iterations in between consecutive distributed optimization rounds. More specifically, each remaining agent $v_j\in\set{R}_k$ initiates a new max-consensus procedure for each optimization round $k\geq\nnint$, after receiving the acknowledgement messages from its out-neighbors. Let the max-consensus iteration index be defined by $\ell=0,1,\ldots,\Lambda$. The max-consensus initialization for each agent $v_j\in\set{R}_k$ is given by
\begin{equation}\label{eq:maxconsensus_init}
\overline{h}_{j,0} = \max_{v_l \in \outneighbor{j}} \Big\{ \alpha_{l,k-1}(1 - \alpha_{l,k}) \Big\},
\end{equation}
where $\alpha_{l,k}=1$ is broadcast by $v_l\in\set{V}_k$ at time $k$, and $\alpha_{l,k}=0$ if no message is received from $v_l$. Then, for each iteration $\ell=0,1,\ldots,\Lambda$ agent $v_j$ broadcasts $\overline{h}_{j,\ell}$, receives $\overline{h}_{i,\ell}$ from its currently active in-neighbors $v_i\in\inneighbor{j,k}$, and updates its local detection variable as 
\begin{align}\label{eq:maxconsensus_itr}
    \overline{h}_{j,\ell+1} &= \max_{v_i \in \inneighbor{j,k}}\Big\{\overline{h}_{i,\ell}, \overline{h}_{j,\ell}\Big\}.
\end{align}
By the end of $\Lambda=\overline{D}$ max-consensus iterations, all the active agents within the cluster become aware of the departure, if any, and they set $h_{j,k}=\overline{h}_{j,\Lambda}$. Note that, $\overline{D}$ is assumed to be known by the agents, although one can employ distributed mechanisms to compute $D$ as in \cite{oliva2016distributed}.

Algorithm~\ref{alg:OpenRC} provides a view of the proposed distributed algorithm for an active agent $v_j\in\set{V}_{k+1}$.
\begin{algorithm}[]
\caption{-- \textsc{Open-GT} iterations at agent $v_j\in\set{V}_{k+1}$.}
\label{alg:OpenRC}
{	\begin{algorithmic}[1]
    \Statex \textbf{Input:} Step-size $\gamma$ and diameter upper-bound $\overline{D}$
    \Statex \textbf{for} $k\in\nnint$\textbf{:}
    \Statex\quad \textbf{if} $v_j\in\mathcal{J}_{k+1}$ (arriving)
    \Statex\quad\quad \textbf{Initialize:} $x_{j,k+1}$, $y_{j,k+1}$, $w_{j,k+1}$, $h_{j,k+1}$ via \eqref{eq:opengt_init}
    \Statex\quad \textbf{else if} $v_j\in\mathcal{R}_{k+1}$ (remaining)
    \Statex\quad\quad \parbox[t]{210pt}{\textbf{Assign weights:} $c_{lj,k}$ via \eqref{eq:c-weights}\strut}
    \Statex\quad\quad \parbox[t]{210pt}{\textbf{Coordinate departure information:} determine $h_{j,k}$ via \eqref{eq:maxconsensus_init}, \eqref{eq:maxconsensus_itr}\strut}
    \Statex\quad\quad \parbox[t]{210pt}{\textbf{Receive:} $c_{ji,k}(x_{i,k}-\gamma w_{i,k})$, $c_{ji,k}y_{i,k}$, $c_{ji,k} w_{i,k}$ and $h_{i,k}$ from each in-neighbor $v_i\in\inneighbor{j,k}$\strut}
    \Statex\quad\quad \parbox[t]{210pt}{\textbf{Transmit:} $c_{lj,k}(x_{j,k}-\gamma w_{j,k})$, $c_{lj,k}y_{j,k}$, $c_{lj,k}w_{j,k}$ and $h_{j,k}$ to each out-neighbor $v_l\in\outneighbor{j,k}$\strut}
    \Statex\quad\quad \parbox[t]{210pt}{\textbf{Update:} $x_{j,k+1}$, $y_{j,k+1}$, and $w_{j,k+1}$ via \eqref{eq:opengt_update}\strut}
    \Statex\quad \textbf{end if}
    \Statex \textbf{Output:} $z_{j,k+1}=x_{j,k+1}/y_{j,k+1}$.
	\end{algorithmic}
}
\end{algorithm}

For analysis purposes, we proceed by stacking the variables of all potentially active agents in the maximal network $\set{G}=(\set{V},\set{E})$, into the column vectors
\begin{align}
    \vect{x}_k &:= \begin{bmatrix}
        x_{1,k},~x_{2,k},~\ldots,~x_{n,k}
    \end{bmatrix}^{\top}
\end{align}
and respectively for $\vect{y}_k$, $\vect{z}_k$, $\vect{w}_k$, and $\vect{g}_k$. Moreover, we stack all agents initial states when joining into $\widehat{\vect{x}}_k:=\begin{bmatrix}  \widehat{x}_{1,k},~\widehat{x}_{2,k},~\ldots,~\widehat{x}_{n,k}
    \end{bmatrix}^{\top}$. Further defining the arrival trigger for agent $v_j\in\set{J}_k$ as $\psi_{j,k}:=\alpha_{j,k}(1-\alpha_{j,k-1})$, and stacking for each agent we have $\grvect{\psi}_k:=\begin{bmatrix}
        \psi_{1,k},~\psi_{2,k},~\ldots,~\psi_{n,k}
    \end{bmatrix}^{\top}$. Similarly, we can stack the coordinated departure information of all agents into $\grvect{\eta}_k:=\begin{bmatrix}
        h_{1,k},~h_{2,k},~\ldots,~h_{n,k}
    \end{bmatrix}^{\top}$. Now we can define the matrices $H_k := \operatorname{diag}(\grvect{\eta}_k)$ and $\Psi_k := \operatorname{diag}(\vect{\grvect{\psi}}_k)$. Moreover, for analysis, one can collect all the push-weights to form the nonnegative matrix $C_k:=\{c_{lj,k}\} \in \mathbb{R}_{\geq0}^{n \times n}$, whose columns that correspond to active agents sum up to $1$, while all the entries that correspond to inactive agents are $0$. Now, the equivalent vector-matrix representation of the proposed algorithm for all potentially active agents is given by:
\begin{subequations}\label{eq:opengt}
    \begin{align}
    \vect{x}_{k+1} &= (I-\Psi_{k+1})C_k(\vect{x}_{k} - \gamma \vect{w}_{k}) + \Psi_{k+1} \widehat{\vect{x}}_{k+1},\\
    \vect{y}_{k+1} &= (I-\Psi_{k+1})C_k \vect{y}_{k} + \Psi_{k+1}\mathbf{1},\\
    \vect{z}_{k+1} &= (\operatorname{diag}(\vect{y_{k+1}))^{-1}} \vect{x}_{k+1},\\
    \vect{w}_{k+1} &= (I-H_k) (C_k \vect{w}_{k} - \vect{g}_k) + \vect{g}_{k+1}.
    \end{align}
\end{subequations}

\section{Analysis}
In the OMAS setting, explicit convergence cannot be guaranteed due to the perturbations introduced by agent arrivals and departures, which continually shift cluster minimizers. Instead, we analyze the \emph{eventual~correctness} of the iterations in \eqref{eq:opengt}, by examining whether agents' estimates track the evolving minimizer of their current cluster. Assuming that the OMAS becomes stable after some finite time $k^{\prime}$, we study the convergence behavior during this stable period.

The following lemma shows that the sum of the local $w$-iterates in \eqref{eq:opengt_update_w} is equal to the sum of the local gradients $g_{j,k}=\nabla f_{j,k}(z_{j,k})$ at any time $k\in\nnint$, using the column stochasticity property of the weights $c_{ij,k}$ that correspond to active agents in any given cluster $q\in\set{Q}$.
\begin{lemma}{}\label{lemma:1}
    Let Assumptions~\ref{ass:3} and \ref{ass:4} hold. Then the following holds for each cluster $q\in\set{Q}$ at any $k\in\set{K}^q$
    \begin{align}\label{lem:lem1}
        \sum_{v_j\in\set{V}^q_{k}} w_{j,k} = \sum_{v_j\in\set{V}^q_{k}} g_{j,k}.
    \end{align}
    \begin{proof}
        We prove the lemma by induction. First observe that, $\sum_{v_j\in\set{J}^q_{k}}w_{j,k} = \sum_{v_j\in\set{J}^q_k} g_{j,k}$ by the (re)-initialization of $v_j\in\set{J}^q_k$. Now, assuming that \eqref{lem:lem1} holds true for a given $k$, we have to show that it holds true at $k+1$.
        \begin{align*}
            \sum_{v_j\in\set{R}^q_{k+1}} w_{j,k+1} &= \sum_{v_j\in\set{R}^q_{k+1}} \left( \sum_{v_i\in\set{V}^q_{k}} c_{ji,k} v_{i,k} + g_{j,k+1} - g_{j,k} \right) \\ 
            &= \sum_{v_j\in\set{V}^q_{k}} w_{j,k} +\!\!\! \sum_{v_j\in\set{R}^q_{k+1}} g_{j,k+1} - \!\!\!\sum_{v_j\in\set{R}^q_{k+1}} g_{j,k}\\
            &= \sum_{v_j\in\set{V}^q_{k}} g_{j,k} + \!\!\!\sum_{v_j\in\set{R}^q_{k+1}} g_{j,k+1} - \!\!\!\sum_{v_j\in\set{R}^q_{k+1}} g_{j,k}\\
            &= \sum_{v_j\in\set{R}^q_{k+1}} g_{j,k+1}.
        \end{align*}
    Thus, since
    \begin{align*}
        \sum_{v_j\in\set{V}^q_{k+1}} w_{j,k+1} &= \sum_{v_j\in\set{R}^q_{k+1}} w_{j,k+1} + \sum_{v_j\in\set{J}^q_{k+1}} w_{j,k+1}\\
        &= \sum_{v_j\in\set{V}^q_{k+1}} g_{j,k+1},
    \end{align*}
    the proof is complete.
    \end{proof}
\end{lemma}

Now suppose that at some $k^{\prime}$, the arrivals and departures of agents are stopped. The formed clusters are now fixed in terms of agents, and they include only the remaining agents $v_j\in\set{R}_{k^{\prime}}$. Hence, for $k \geq k^{\prime}$ we have $\set{V}_k=\set{V}=\set{R}_{k^{\prime}}$. In what follows we analyze the convergence of the remaining agents for each cluster independently. Let us further define the variables of the remaining agents in each cluster $q\in\set{Q}$. Let $\vect{x}_{k+1}^q$, $\vect{y}_{k+1}^q$, $\vect{z}_{k+1}^q$, and $\vect{w}_{k+1}^q$ the stack of variables maintained by the agents in cluster $q\in\set{Q}$ during $k\in\set{K}^q$. 

The vector-matrix form of the iterations of the agents that remain in cluster $q$ for $k\geq k^{\prime}$ is now reduced to the one of the Push-DIGing algorithm in \cite{nedic2017achieving} for closed networks with time-varying links. Hence the iterations in \eqref{eq:opengt} for the cluster $q$ and for $k\geq k^{\prime}$, gives $\Psi_{k}^q=I$ and $H_{k}^q=I$, and the iterations of the agents in the cluster become:
\begin{subequations}\label{eq:push-diging}
    \begin{align}
    \vect{x}_{k+1}^q &= C_k^q (\vect{x}_{k}^q - \gamma \vect{w}_{k}^q),\\
    \vect{y}_{k+1}^q &= C_k^q \vect{y}_{k}^q,\\
    \vect{z}_{k+1}^q &= (\operatorname{diag}(\vect{y_{k+1}^q))^{-1}} \vect{x}_{k+1}^q,\\
    \vect{w}_{k+1}^q &= C_k^q \vect{w}_{k}^q + \vect{g}_{k+1}^q - \vect{g}_k^q, 
    \end{align}
\end{subequations}
where $C_k^q$ is a column-stochastic matrix with weights assigned as in \eqref{eq:c-weights} for the cluster $q$. Note that the gradients $\vect{g}_{k+1}^q$ and $\vect{g}_{k}^q$ are evaluated at $\vect{z}_{k+1}^q$ and $\vect{z}_{k}^q$, respectively.

\begin{theorem}
    Assume that there exists a time $k^{\prime}$ when agents stop arriving and/or departing the network. Then, the \textsc{Open-GT} algorithm in \eqref{eq:opengt} is reduced in multiple instances of \eqref{eq:push-diging} for each formed cluster $q\in\set{Q}$, each converging to $x^{q\ast}$ asymptotically as $k\rightarrow\infty$, by choosing a sufficiently small gradient step-size $\gamma>0$.
\end{theorem}
\begin{proof}
Follows similar arguments \emph{mutatis mutandis} to those in \cite{nedic2017achieving}, invoking Lemma~\ref{lemma:1} and Assumptions~\ref{ass:1}-\ref{ass:4}.
\end{proof}

\section{Simulation Results}
In this section, we evaluate the proposed algorithm through numerical simulations. We consider an OMAS with a maximal topology shown in Fig.~\ref{fig:simulation_digraph}, where only agents $v_3$, $v_4$, and $v_5$ are allowed to join or leave at runtime. Notably, if $v_4$ departs at any time $k \in \mathbb{Z}_{\geq0}$, the network splits into two strongly connected clusters, regardless of the presence of $v_3$ or $v_5$. Each joining agent $v_j \in \set{J}_k$ initializes its state as $x_{j,k} = \widehat{x}_{j,k} \in [1,5]$ (chosen uniformly at random), $y_{j,k} = 1$, $w_{j,k} = \nabla f_{j,k}(x_{j,k})$, and $h_{j,k} = 0$.
\begin{figure}[h!]
\centering
    \includegraphics[scale=0.85]{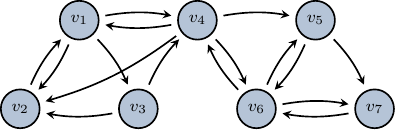}
    \caption{The maximal directed network topology $\set{G}=(\set{V},\set{E})$ consisting of $n=7$ agents, for an OMAS described by the digraph $\set{G}_k=(\set{V}_k,\set{E}_k)$.}
    \label{fig:simulation_digraph}
\end{figure}

In Fig.~\ref{fig:simulation_results} we present the evolution of the OMAS considered in Fig.~\ref{fig:simulation_digraph} when active agents execute the \textsc{Open-GT} algorithm. The top plot shows the activation vector $\vect{\alpha}_k$, indicating when each agent $v_j \in \set{V}_k$ is active, with different colors representing dynamically formed clusters. The middle plot illustrates the estimate $z_{j,k}$ of each active agent, potentially belonging to different clusters. The bottom plot shows the optimality error $\lVert \vect{z}_k - \vect{1}x_k^{q\ast} \rVert_2$ for each cluster $q \in \set{Q}$. When agent $v_4$ is active, the network becomes connected, forming a single cluster with a global minimizer, as seen in the intervals $[1,80]$ and $[310,420]$. In the latter, the minimizer shifts due to the departure of agent $v_5$ and re-arrival of agent $v_3$. When $v_4$ is inactive, the network splits into multiple clusters, each with its own minimizer $x_k^{q\ast}$. Throughout, agents exchange information and update their variables via \eqref{eq:opengt_update} with local neighbors, aiming to minimize their cluster objective. As expected from Theorem~1, during stable periods (\ie no arrivals or departures), the optimality error decreases and agent estimates $z_{j,k}$ converge toward the cluster minimizer $x_k^{q\ast}$. Transient spikes in $z_{j,k}$ may occur due to the influence of arriving or departing agents, particularly those with high connectivity or strong local cost contributions.

\begin{figure}[h!]
    \includegraphics[scale=0.82]{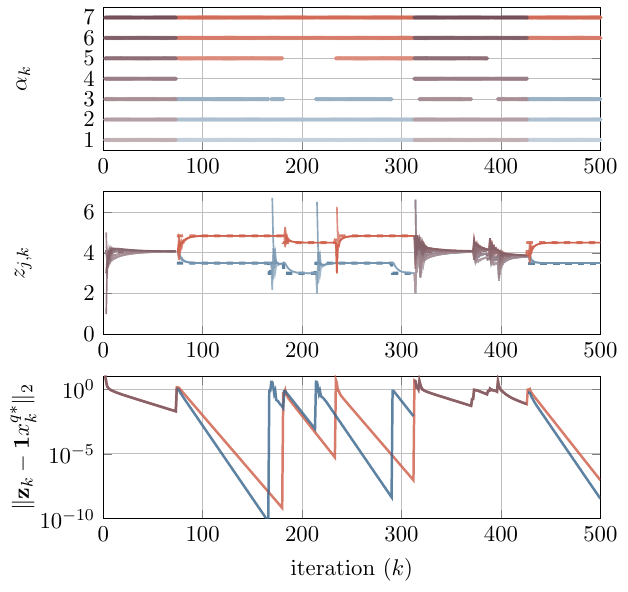}
    \caption{Agents' evolution executing the \textsc{Open-GT} algorithm.}
    \label{fig:simulation_results}
\end{figure}

\section{Conclusions}
In this work, we addressed the distributed optimization in open multi-agent systems with directed, time-varying communication. We proposed \textsc{Open-GT}, a novel algorithm that enables agents to adapt to dynamic topology changes using acknowledgment messages to detect active out-neighbors and a max-consensus mechanism to propagate departure information. Our preliminary analysis shows that \textsc{Open-GT} maintains consistency within dynamically formed clusters and converges to the cluster-wide minimizer once cluster membership stabilizes. Simulations further demonstrate the algorithm’s robustness to network variations, highlighting its practical effectiveness in dynamic OMAS environments.

\bibliographystyle{IEEEtran}
\bibliography{references}

\end{document}